\documentclass[a4paper, 12pt]{article}
\usepackage{amsfonts}
\usepackage{amsmath}
\usepackage{graphicx}
\usepackage{todonotes,hyperref,amsthm, authblk}
\usepackage{algorithm}
\usepackage[framemethod=TikZ]{mdframed}

\newtheorem{clm}{Claim}
\usepackage{mathtools}

\usepackage[T1]{fontenc}
\usepackage[sc,osf]{mathpazo}

\newcommand{\TSSDOM}{{\sc TSS-D1M}}

\newcommand{\defparproblem}[4]{
  \vspace{3mm}
\noindent\fbox{
  \begin{minipage}{.95\textwidth}
  \begin{tabular*}{\textwidth}{@{\extracolsep{\fill}}lr} \textsc{#1}  & {\bf{Parameter:}} #3 \\ \end{tabular*}
  {\bf{Input:}} #2  \\
  {\bf{Question:}} #4
  \end{minipage}
  }
  \vspace{2mm}
}

\mdfdefinestyle{MyFrame}
{
linecolor=blue,
outerlinewidth=1.0pt,
roundcorner=20pt,
innertopmargin=\baselineskip,
innerbottommargin=\baselineskip,
innerrightmargin=20pt,
innerleftmargin=20pt,
backgroundcolor=gray!50!white
}

\theoremstyle{definition}

\theoremstyle{plain}
\newtheorem{theorem}{Theorem}
\newtheorem{lemma}[theorem]{Lemma}

\newtheorem{observation}[theorem]{Observation}
\newtheorem{proposition}[theorem]{Proposition}

\theoremstyle{remark}
\newtheorem{remark}[theorem]{Remark}

\begin{document}
\title{{\sc Target Set Selection} Parameterized by Vertex Cover and More}
%
%
\author[1]{Suman Banerjee}
\author[2]{Rogers Mathew}
\author[2]{Fahad Panolan}
\affil[1]
{
Indian Institute of Technology Jammu, \authorcr 
Jammu \& Kashmir 181221, India. \authorcr
suman.banerjee@iitjammu.ac.in
}
\affil[2]
{
Department of Computer Science and Engineering,\authorcr
Indian Institute of Technology Hyderabad, India.\authorcr
\{rogers, fahad\}@cse.iith.ac.in
}


\maketitle              

\begin{abstract}
Diffusion is a natural phenomenon in many real-world networks. Spreading of ideas, rumours in an online social network; propagation of virus, wormhole in a computer network; spreading of diseases in a human contact network, etc. are some real-world examples of this.  Diffusion  often starts from a set of initial nodes known as \emph{seed nodes}. A node can be in any one of the following two states: \emph{influenced} (\emph{active}) or \emph{not influenced} (\emph{inactive}). We assume that a node can change its state from inactive to active, however, not vice versa. Only the seed nodes are active initially and the information is dissipated from these seed nodes in discrete time steps. Each node $v$ is associated with a threshold value $\tau(v)$ which is a positive integer. A node $v$ will be influenced at time step $t$, if there are at least $\tau(v)$ number of nodes in its neighborhood which have been activated on or before time step $t-1$. The diffusion stops when no more node-activation is possible. 

Given a simple, undirected graph $G$ with a threshold function $\tau:V(G) \rightarrow \mathbb{N}$, the \textsc{Target Set Selection} (TSS) problem is about choosing a minimum cardinality set, say $S \subseteq V(G)$, such that starting a  diffusion process with $S$ as its seed set will eventually result in activating all the nodes in $G$. For any non-negative integer $i$, we say a set $T\subseteq V(G)$ is a \emph{degree-$i$ modulator} of $G$ if the degree of any vertex in the graph $G-T$ is at most $i$. Degree-$0$ modulators of a graph are precisely its vertex covers. Consider a graph $G$ on $n$ vertices and $m$ edges. We have the following results on the TSS problem:
\begin{itemize}
\item It was shown by Nichterlein et al. [Social Network Analysis and Mining, 2013] that it is possible to compute an optimal-sized target set in $O(2^{(2^{t}+1)t}\cdot m)$ time, where $t$ denotes the  cardinality of a minimum degree-$0$ modulator of $G$. We improve this result by designing an algorithm running in time 
$2^{O(t\log t)}n^{O(1)}$.

\item We design a $2^{2^{O(t)}}n^{O(1)}$ time algorithm to compute an optimal target set for $G$, where $t$ is  the size of a minimum degree-$1$ modulator of $G$.  
\end{itemize}
\end{abstract}

\section{Introduction}
Diffusion is a natural phenomenon in many real\mbox{-}world networks such as diffusion of information, innovation, ideas, rumors in an online social network 
\cite{bakshy2012role};
propagation of virus, wormhole in a computer network \cite{hu2006wormhole}; spreading of contaminating diseases in a human contact network \cite{salathe2010high}, and many more. Depending on the situation, we want to maximize/minimize the spread. For example, in the case of propagation of information in a social network, sometimes we want to maximize the spread so that a large number of people are aware of the piece of information. On the other hand, in the case of spreading of contaminating diseases, we would want to minimize the spread. In this paper, the practical essence of our study is in and around the first situation. 
\par Diffusion starts from a set of initial nodes known as \emph{seed nodes}. A node can be in any one of the following two states: \emph{influenced} (also known as \emph{active}) or \emph{not influenced} (also known as \emph{inactive}). We assume that a node can change its state from inactive to active, however, not vice versa. Only the seed nodes are active initially and the information is disseminated in discrete time steps from these seed nodes. Each node $v$ is associated with a threshold value $\tau(v)$ which is a positive integer. A node $v$ will be influenced at time step $t$, if it has at least $\tau(v)$ number of nodes in its neighborhood which have been activated on or before time step $(t-1)$. The diffusion process stops when no more node\mbox{-}activation is possible. A set of seed nodes is called a \emph{target set} if diffusion starting from these seed nodes spreads to the entire network thereby influencing every node. 

\medskip
\noindent
{\bf Problem definition.}
In our study, we assume that the social network is represented by an undirected graph 
$G$, where $V(G)$ and $E(G)$ are the set of vertices and edges of $G$, respectively, and there is a threshold function $\tau:V(G)\rightarrow \mathbb{N}$ that assigns a threshold value to each node. Let $S \subseteq V(G)$ be a set of seed nodes from where diffusion starts. As described in the above paragraph, influence propagates in discrete time steps, i.e., $\mathcal{A}[S,0] \subseteq \mathcal{A}[S,1] \subseteq  \mathcal{A}[S,2] \subseteq \dots \subseteq \mathcal{A}[S,i] \subseteq \dots \subseteq V(G)$, where $\mathcal{A}[S,i]$ denotes the set of nodes that has been influenced on or before the $i^{th}$ time stamp and $\mathcal{A}[S,0]= S$. For all $i>0$, the diffusion process can be expressed by the following equation:
\begin{center}
$\mathcal{A}[S,i]=\mathcal{A}[S,i-1] \cup \{u ~:~ \mid N(u) \cap \mathcal{A}[S,i-1] \mid \geq \tau(u)  \}$,
\end{center} where $N(u)$ denotes the set of neighbors of $u$. 
For any seed set $S$, we define ${\sf influence}_G(S) := \underset{t \geq 0} {\bigcup} \mathcal{A}[S,t]$. Observe that $t\leq |V(G)|$ as at least one new node is activated in every time step; else the diffusion process stops. The {\sc Target Set Selection (TSS)} problem is about finding a minimum cardinality target set. In other words, it is about finding a minimum cardinality seed set $S$ such that ${\sf influence}_G(S) = V(G)$.   

\medskip
\noindent
{\bf Related work.}
Chen  \cite{chen2009approximability} showed that TSS cannot be approximated within a factor of $O(2^{\log^{1-\epsilon}n})$ of the optimum for a fixed constant $\epsilon > 0$, unless ${NP} \subseteq {DTIME}(n ^{polylog(n)})$, by a reduction from the MINREP Problem. They also showed that TSS is NP\mbox{-}hard for bounded\mbox{-}degree bipartite graphs with a threshold value not greater than 2 at each vertex by a reduction from a variant of the $3$-SAT problem. For trees, they proposed a polynomial\mbox{-}time exact algorithm. 
Chiang et al.~\cite{chiang2013some} showed that TSS  can be solved in linear time for block\mbox{-}cactus graphs with an arbitrary threshold and for chordal graph with threshold at most 2. Cicalese et al. \cite{cicalese2014latency} proposed exact algorithms to find minimum cardinality target set for bounded clique width and trees.
Chopin et al. \cite{chopin2014constant} showed that upper bounding the threshold to a constant leads to efficiently solvable instances of the TSS problem under the parameterized complexity theoretic framework. They showed that TSS  is W[1]\mbox{-}hard with respect to  the parameters \emph{feedback vertex cover}, \emph{distance to co\mbox{-}graph}, \emph{distance to interval graph}, \emph{pathwidth}, \emph{cluster vertex deletion number}, W[2]-hard with respect to the parameter \emph{seed set cardinality}, and fixed parameter tractable with respect to the parameters \emph{distance to clique} and \emph{bandwidth}. Dvo{\v{r}}{\'a}k  et al. \cite{DBLP:conf/isaac/DvorakKT18} added a few more results in the parameterized setting. They showed that TSS  is W[1]-hard with respect to parameter \emph{neighborhood diversity}, and under majority threshold (i.e., the threshold of each vertex is equal to half of its degree) this problem has an FPT algorithm with respect to the parameters \emph{neighborhood diversity}, \emph{twin cover}, \emph{modular width}.  Bazgan et al. \cite{bazgan2014parameterized} showed that for any functions $f$ and $\rho$ this problem cannot be approximated within a factor of  $\rho(k)$ in $f(k) \cdot n^{O(1)}$ time unless $FPT=W[P]$ even for constant and majority thresholds, where $k$ denotes the cardinality of an optimal target set. Nichterlein et al. \cite{nichterlein2013tractable} showed that for diameter two split graphs TSS  remains W[2]\mbox{-}hard with respect to the parameter \emph{size of the target set}. Also, TSS  is fixed parameter tractable when parameterized by the {\em vertex cover number} and {\em cluster editing number}. Hartmann \cite{hartmann2018target} showed that the TSS Problem is FPT when parameterized with the combined parameters \emph{clique-width} and \emph{maximum threshold value} of the input graph. 
Notice that {clique-width} is smaller than many parameters like vertex cover number and treewidth. 
Bliznets et al. \cite{bliznets_et_al:LIPIcs:2019:10223} presented several faster\mbox{-}than\mbox{-}trivial algorithms under several threshold models such as constant thresholds, dual constant thresholds where the threshold value of each vertex is bounded by one third of its degree. Recently, Keiler et al. \cite{keiler2020target} studied a variant of TSS Problem where the goal is to select a target set that maximizes the diffusion rounds. They showed that the problem is FPT when parameterized by diffusion rounds and maximum threshold if the input graph has bounded local treewidth. The problem is NP-Complete even when we need the diffusion rounds to be at least $4$ and maximum threshold is equal to 2. They also showed that the problem is NP-Hard for planar graphs, $W[1]$-Hard parameterized by treewidth, and polynomial-time solvable on trees.

\medskip
\noindent
{\bf Our contribution.}
For any non-negative integer $i$, we say that a set $T\subseteq V(G)$ is a \emph{degree-$i$ modulator} of $G$ if the degree of any vertex in the graph $G-T$ is at most $i$. Degree-$0$ modulators of a graph are precisely its \emph{vertex covers}. Consider a graph $G$ on $n$ vertices and $m$ edges. We have the following results on TSS. 
\begin{itemize}
\item Nichterlein et al. [Social Network Analysis and Mining, 2013] showed 
that it is possible to compute an optimal-sized target set in $O(2^{(2^{t}+1)t}\cdot m)$ time, where $t$ denotes the  cardinality of a minimum degree-$0$ modulator (i.e., vertex cover) of $G$. In Section \ref{Sec:3}, we improve this result by designing an algorithm that computes an optimal-sized target set in  $2^{O(t\log t)}n^{O(1)}$ time.

\item In Section \ref{Sec:D1M}, we design a $2^{2^{O(t)}}n^{O(1)}$ time algorithm to compute an optimal target set, where $t$ is the size of a minimum degree-$1$ modulator of $G$.  

\end{itemize}

\smallskip
\noindent
{\bf Lower bounds.} 
%
%
Ben-Zwi et al. \cite{ben2011treewidth} proved algorithmic lower bounds for {\sc TSS} when parameterized by the treewidth of the input graph. Towards that, in Lemma  4.3, the authors give a polynomial-time reduction that maps an instance $(G,k)$ of the \textsc{Multi-Colored Clique} problem\footnote{Given a set of $k$ distinct colors and a graph on $n$ vertices whose every vertex is colored with one of the $k$ colors, the \textsc{Multi-Colored Clique} problem is about finding a clique of size $k$ in the graph such that no two vertices of the clique are of the same color.} to an instance $(G', \tau, k')$ of the \textsc{Target Set Selection} (\textsc{TSS}) problem, where $k' = k + {k \choose 2}$, such that $G$ has a multicolored clique of size $k$ if and only if $G'$  has a target set of size $k'$. The \emph{star-deletion number} (minimum number of vertices who removal results in a graph that is a disjoint collection of star graphs) of the graph $G'$ constructed  is $O(k^2)$. Thus, the treewidth and tree-depth  of $G'$ is $O(k^2)$. This helps Ben-Zwi et al. \cite{ben2011treewidth} in proving that for a graph on $n$ vertices  of treewidth (or tree-depth, or star-deletion number) $s$,  the \textsc{TSS} problem cannot be solved in $f(s)n^{o(\sqrt{s})}$ time unless 
the Exponential Time Hypothesis (ETH) fails. 
For definitions of \emph{treewidth} and \emph{tree-depth}, see \cite{cygan2015parameterized}.

It was shown in \cite[Corollary 6.3]{marx2007can} that if the \textsc{Partitioned Subgraph Isomorphism} (\textsc{PSI}) problem\footnote{In {\sc PSI} we are given two graphs $G$ and $H$, a bijection $f_{G}:V(G) \longrightarrow [\ell]$, and a function $c_{H}:V(H) \longrightarrow [\ell]$, where $|V(G)|=\ell$. The objective is to test the existence of a  subgraph isomorphism $\phi$ from $G$ to $H$ such that for all $v \in V(G)$, $f_{G}(v)=c_{H}(\phi(v))$.} can be solved in $f(G)n^{o(\frac{k}{\log k})}$ time, then ETH fails, where $f$ is an arbitrary function, $n=|V(H)|$, $G$ is connected, and $k$ is the number of edges of the smaller graph $G$. 
%
%
We observe that by making small modifications to the reduction given in \cite{ben2011treewidth}, we can give a polynomial-time reduction from an instance $(G, H, f_G, C_H)$ of the PSI problem, where $|V(G)|=\ell, |E(G)|=k$, and $|V(H)| = n$, to an instance $(G', \tau, k')$ of the \textsc{TSS} problem, where $k' = \ell + k$, such that $(G, H, f_G, C_H)$ is a YES-instance of the PSI problem if and only if $G'$ has a target set of size $k'$. The star-deletion number of the graph $G'$ constructed can be shown to be $O(k)$. This helps us in improving the lower bound given in \cite{ben2011treewidth}. We can show that for a graph on $n$ vertices  of treewidth (or tree-depth, or star-deletion number) $s$,  the \textsc{Target Set Selection} problem cannot be solved in $f(s)n^{o(\frac{s}{\log s})}$ time unless ETH is false. The reduction is similar to the one in \cite{ben2011treewidth}, where we use {\sc PSI} instead of {\sc Multi-Colored Clique}. Hence we omit 
the proof here.

\subsection{Preliminaries}
Throughout the paper, we consider finite, undirected and simple graphs. For any vertex $v$ in a graph $G$, we shall use $d_G(v)$ to denote the number of edges incident on $v$ and $N_G(v)$ to denote the set of vertices adjacent to $v$. 
We omit the subscript $G$ if the graph $G$ is clear from the context. 
For any subset, say $S$, of the set of vertices of the graph under consideration, we shall use $N_S(v)$ to denote the set of neighbors of $v$ in $S$. For a graph $G$, we shall use $V(G)$ and $E(G)$ to denote its vertex set and edge set, respectively. For a set $S \subseteq V(G)$, we shall use (i) $G[S]$ to denote the subgraph induced by $S$ on $G$, and (ii) $G-S$ to denote the graph $G[V(G)\setminus S]$. 
A set $S\subseteq V(G)$ is said to be an \emph{independent set} of $G$ if no two vertices in $S$ are adjacent with each other. A set $U \subseteq V(G)$ is a {\em vertex cover} of $G$, if $V(G)\setminus U$ is an independent set. 
A vertex cover with minimum cardinality is called an optimal vertex cover. 
For a function $f \colon X \mapsto Y$, $X'\subseteq X$ and $Y'\subseteq Y$, $f(X')=\{f(x)~\colon~x\in X'\}$ and 
$f^{-1}(Y')=\{x\in X~\colon~f(x)\in Y'\}$.


\section{FPT algorithm  parameterized by degree-0 modulator (vertex cover number)} \label{Sec:3}
Consider the target set selection problem on a graph $G$ with a threshold function $\tau:V(G) \rightarrow \mathbb{N}$.  Clearly, all the vertices $v$ having $\tau(v) > d_G(v)$ are present in every feasible target set. So while designing an algorithm to find an optimal target set in $G$, it is safe to include all such vertices into the solution set to be constructed. Throughout this section we therefore assume that the graph $G$ under consideration has $\tau(v) \leq d_G(v)$, for all $v \in V(G)$. 
We begin by stating the following easy-to-see remark. 
\begin{remark} \label{rem:vertex_cover_target_set}
Let $C$ be a vertex cover of a graph $G$ with a threshold function $\tau:V(G) \rightarrow \mathbb{N}$. Then, $C$ is a target set in $G$. 
\end{remark}

\begin{lemma}
\label{Lemma:2}
Let $t$ be the size of an optimal vertex cover in a graph $G$ with a threshold function $\tau:V(G) \rightarrow \mathbb{N}$. Then, the diffusion process starting from any non-empty seed set $S \subseteq V$ terminates in at most $2t$ rounds. 
\end{lemma}
\begin{proof}
Let $C$ be an optimal vertex cover of $G$ of size $t$ and let $B=V(G)\setminus C$. We know that $B$ is an independent set. Let $S_i$ be the set of uninfluenced nodes that were influenced in Round $i$ of the diffusion process. We have $S_0 = S$. Assume the diffusion process terminates in $k$ rounds. For each $0 \leq i \leq k$, observe that $S_i$ is a non-empty set. For $0 \leq i < k$, since $B$ is an independent set, it is not possible to have both $S_i$ and $S_{i+1}$ to be subsets of $B$. Thus, in every two consecutive rounds, at least one uninfluenced vertex from $C$ will be influenced. This implies that $C\cap {\sf influence}_G(S)$ will be influenced in at most $2t-1$ steps. 
Therefore, the diffusion process will terminate in at most $2t$ steps. 
\end{proof}

TSS when parameterized by degree-$i$ modulator is defined below.

\defparproblem{\textsc{TSS-D$i$M}}{An undirected graph $G$ on $n$ vertices, a threshold function $\tau: V(G) \mapsto \mathbb{N}$, $k \in \mathbb{N}$, and a degree-$i$ modulator $C$ of size $t$.}{t}{Is there a target set of size $k$?}

Recall that \textsc{TSS-D$0$M} is the TSS problem parameterized by vertex cover. 
Toward getting our FPT algorithm, 
we give a Turing reduction from \textsc{TSS-D$0$M} 
to a variant of the hitting set problem  
which is defined below. 

%

\defparproblem{\textsc{Multi\mbox{-}Hitting Set}}{A universe $U$, where $|U| \leq n$, a collection of subsets $S_1, S_2, \ldots, S_t \subseteq U$ and $q, l_1, l_2, \ldots, l_t \in \mathbb{N}$ such that $l_j \leq t$, for all $j \in [t]$.}{t}{Is there a subset $H \subseteq U$, such that $|H| \leq q$ and $|H \cap S_i| \geq l_i$, for all  $i \in [t]$.}

%



Now, we describe our Turing reduction that constructs $2^{O(t\log t)}$ instances of \textsc{Multi\mbox{-}Hitting Set} from a given instance of {\sc TSS-D$0$M}.

\begin{theorem} \label{thm:TSS-MHS}
There is an algorithm that given an instance $(G, \tau, k, C)$ of  {\sc TSS-D$0$M}, where $t=\vert C\vert$, runs in $2^{O(t \log t)}\cdot n^{O(1)}$ time, and outputs a collection of instances $\mathcal{I}=\{I^j = (V(G) \setminus C, q^j, S_{1}^{j}, S_{2}^{j}, \ldots, S_{t}^{j},l_{1}^{j}, l_{2}^{j}, \ldots, l_{t}^{j}): j \in [s]\}$ of {\sc Multi-Hitting Set} such that the following holds:
\begin{itemize}
\item (a) The number of instances, i.e., $s \leq 2^{O(t \log t)}$. 
\item (b) $(G,\tau, k,C)$ is a YES-instance of {\sc TSS-D$0$M} if and only if there exists $ j \in [s]$ such that 
$I^j$ is a YES-instance of \textsc{Multi\mbox{-}Hitting Set}.
\end{itemize}
\end{theorem}

\begin{proof}
First we describe the construction of $2^{O(t \log t)}$ many instances of {\sc Multi-Hitting Set} from a given instance $(G, \tau, k, C)$ of {\sc TSS-D$0$M}. 
\paragraph*{Construction.} \label{Cons:3}
Consider the given {\sc TSS-D$0$M} instance  $(G, \tau, k, C)$, where $G$ is a graph on $n$ vertices and $C=\{v_1, \ldots , v_t\}$ is a vertex cover of size $t$. Let $B=V(G) \setminus C$. Note that $B$ is an independent set. For each $v_i\in C$, we guess a time stamp $T(v_i)$ in which $v_i$ will be influenced. From Lemma \ref{Lemma:2}, we know that $T(v_i)\in \{0, \ldots , 2t\}$. There are $t$ vertices in $C$ and each one of them can be assigned any one of these $2t+1$ distinct values. So, there are $(2t+1)^t=2^{\mathcal{O}(t \log t)}$ possible guesses for the time stamps of the vertices in $C$. Now, among all the $2^{\mathcal{O}(t \log t)}$ possibilities, let us consider the $j$-th one. That is, consider that we are given a $t$-tuple $(T^j(v_1), T^j(v_2), \ldots , T^j(v_t))$ of guessed time stamps for vertices in $C$. Based on these guessed time stamps for the vertices in $C$, for any $u \in B$, we compute $T^j(u)$ as  
$1+\min \{x\in {\mathbb N} \colon \mbox{number of vertices in } N(u) \mbox{ with time stamp at most } x, \mbox{ is at least }\tau(u)\}.$
\\ For each $v_i \in C$, 
we define $l_i^j := \max \{0, \tau(v_i) - |\{w \in N(v_i):T^j(w)<T^j(v_i)\}|\}$ and $S^j_i := \{w \in N_B(v_i):T^j(w) \geq T^j(v_i)\}$. We remark that if $l_i^j\geq t$, then we will not include $I_j$ in the collection of the output instances. 
We thus have the $j$-th instance of the multi-hitting set problem where $U = B$, $S^j_i$'s and $l^j_i$'s are as defined above, $t=|C|$, and $q^j = k - |\{v_i \in C:T(v_i) = 0\}|$. This completes the construction of output instances. It is easy to verify that the number of instance in ${\cal I}$ is $2^{O(t\log t)}$. 
Property (b) follows from the following two claims. 

\begin{clm}
\label{clm:Construction_TSS_gives_hitting_set}
Suppose it is given that one of the $t$-tuples we guess, say the $j$-th $t$-tuple $(T^j(v_1), T^j(v_2), \ldots , T^j(v_t))$, happens to represent the activation time of vertices in $C$ corresponding to some feasible target set $S$. Then, for all $i \in [t], |S \cap S^j_i| \geq  l_i^j$. That is, $S \cap B$ is a solution for the  instance $I^j$ of \textsc{Multi\mbox{-}Hitting Set}. 
\end{clm}
\begin{proof}
For any $w\in B$, $N(w)\subseteq C$ and we know that for any $v_i\in C$, $v_i$ is influenced in step $T^j(v_i)$ for the target set $S$. This implies that for any $w\in B\setminus S$, $w$ is influenced in step $T^j(w)$. Therefore, as $v_i$ is influenced in step $T^j(v_i)$, at least $l_i^j$ vertices from $S_i^j$ should be there in the target set $S$. 
This implies that $S \cap B$ is a solution for the  instance $I^j$. 
 %
\end{proof}

\begin{clm}
\label{rem:Construction_hitting_set_gives_TSS}
Let $H$ be a hitting set for the instance $I^j$ of \textsc{Multi\mbox{-}Hitting Set}. Then $S=H \cup \{v_i \in C~:~T^j(v_i)=0\}$ is a target set for $G$. 
\end{clm}
\begin{proof}
To prove the claim it is enough to show that $C\subseteq {\sf influence}_G(S)$. We prove by induction on $q$ that all the vertices $v_i\in C$ with $T^j(v_i)\leq q$ will be influenced by the end of step $q$. 
The base case is when $q\leq 1$. 
Clearly, for $q=0$, all the vertices 
$v_i\in C$ with $T^j(v_i)\leq q$ is influenced initially because $v_i\in S$. Notice that for any $w\in B$, $T^j(w)>0$ and since $H$ is a solution for the instance $I^j$ of \textsc{Multi\mbox{-}Hitting Set}, for any $v_i\in C$ with $T^j(v_i)=1$, at least $l_i^j=\tau(v_i)$ vertices from $N(v_i)$ are there in $H$. This implies that for any $v_i\in C$ with $T^j(v_i)=1$, $v_i$ will be influenced in step 1. Now consider the induction step for which $q>1$.  
Now consider a vertex $v_i\in C$ with $T^j(v_i)=q$. We know that at least $l_i^j$ vertices from $S^j_i := \{w \in N_B(v_i)~:~T^j(w) \geq q\}$ are present in $H$ because $H$ is a solution to $I^j$. 
By the induction hypothesis, we have that for any vertex $v_r\in C$ with $T^j(v_r)=q-2$ is influenced at the end of step $q-2$. This implies that all the vertices in $\{w \in N_B(v_i)~:~T^j(w) < q\}$ are influenced by the end of step $q-1$. 
Therefore, at least $\tau(v_i)$ vertices from $N(u)$ will be influenced by the end of step $q-1$. This implies that $v_i$ will be influenced in step $q$. This completes the proof of the claim. 
\end{proof}
This completes the proof of the theorem.
\end{proof}

Next, we design an FPT algorithm for \textsc{Multi\mbox{-}Hitting Set}. 

\begin{theorem}
\label{thm:multi_hitting_set}
Given an instance $I= (U, q, S_{1}, S_{2}, \ldots, S_{t};l_{1}, l_{2}, \ldots, l_{t}) )$ of \textsc{Multi\mbox{-}Hitting Set}, there is an algorithm of running time $2^{O(t \log t)}\cdot n^{O(1)}$ to solve $I$, where $|U| = n$.  
\end{theorem}
\begin{proof}
Let $U = \{u_1, \ldots , u_n\}$ and let $U_j = \{u_1, \ldots , u_j\},~1\leq j \leq n$. 
We design a dynamic programming algorithm, where in the DP table entry $D_j(q', l_1', l_2', \ldots , l_t')$ we store a hitting set (if one exists; else, it will be equal to NULL) of size at most $q'$ that is a subset of $U_j$ and hits each $S_i$ on at least $l_i'$ elements, where $0 \leq l_i'\leq l_i$, for all $i \in [t]$ and $0 \leq q'\leq q$. 
The case when $j=0$ can be computed easily as follows. 

$$
D_0(q', l_1', l_2', \ldots , l_t') = \left\{ \begin{array}{rl}
\emptyset & \mbox{if }  l_i'=0 \mbox{ for all } i,\\ 
NULL & \mbox{otherwise}. 
\end{array}\right.
$$

We compute the DP table entries in increasing order of $j$. Consider the case when $j \geq 1$. Without loss of generality, assume $u_{j} \in S_1\cap \cdots \cap S_k$ and $u_{j} \notin S_{k+1} \cup \cdots \cup S_t$. Then, 
for any values of $q', l_1', l_2', \ldots , l_t'$ such that $0 \leq l_i'\leq l_i$, for all $i \in [t]$ and $0 \leq q'\leq q$, 
we compute $D_{j}(q', l_1', l_2', \ldots , l_t')$  as follows.  
If $D_{j-1}(q', l_1', l_2', \ldots , l_t')\neq NULL$, then 
$$D_{j}(q', l_1', l_2', \ldots , l_t') = D_{j-1}(q', l_1', l_2', \ldots , l_t').$$
If $D_{j-1}(q'-1, l_1'-1, l_2'-1, \ldots , l_k'-1, l_{k+1}', \ldots , l_t')\neq NULL$, then set 
$$D_{j}(q', l_1', l_2', \ldots , l_t') = D_{j-1}(q'-1, l_1'-1, l_2'-1, \ldots , l_k'-1, l_{k+1}', \ldots , l_t') \cup \{u_j\}.$$
Otherwise, we set 
$D_{j}(q', l_1', l_2', \ldots , l_t') = \mbox{ NULL}$. Using this dynamic programming approach, we eventually compute $D_n(q,l_1, \ldots , l_t)$. Since each of $q, l_1, \ldots , l_t$ is at most $t$, we can compute this in time $(t+1)^t n^{O(1)} = 2^{O(t\log t)}\cdot n^{O(1)}$.   

Next we prove the correctness of our algorithm. Towards that we claim that $D_n(q,l_1, \ldots , l_t) = \mbox{ NULL}$ if and only if $I$ is a NO-instance. We prove this by proving a more general statement. We will show that, for every $j \in \{0,\ldots, n\}$, $D_j(q',l_1', \ldots , l_t') = \mbox{ NULL}$ if and only if $I_j(q',l_1', \ldots , l_t') := (U_j, q', S_{1,j}, \ldots , S_{t,j}; l_1', \ldots , l_t')$ is a NO-instance where $S_{i,j} = S_i \cap U_j$, $l_i' \leq l_i$, for every $i$, and $0 \leq q' \leq q$. We prove this by strong induction on $j$. It is easy to see that the statement is true for the base case  when $j=0$. 
Consider the induction step when $j>0$. Notice that, by the induction hypothesis, the statement is true for all $j'<j$. Without loss of generality,  assume $u_{j} \in S_1\cap \cdots \cap S_k$ and $u_{j} \notin S_{k+1} \cup \cdots \cup S_t$. Suppose $D_j(q',l_1', \ldots , l_t') = \mbox{ NULL}$. This implies, both $D_{j-1}(q'-1,l_1'-1, \ldots , l_k'-1, l_{k+1}', \ldots , l_t') = \mbox{ NULL}$ and $D_{j-1}(q',l_1', \ldots , l_t') = \mbox{ NULL}$. Thus, by induction hypothesis, 
the instances 
$I_{j-1}(q',l_1', \ldots , l_t')$ and  
$I_{j-1}(q'-1,l_1'-1, \ldots,l_k'-1,\ldots, l_t')$ are NO-instances. 
Hence, $I_{j}(q',l_1', \ldots , l_t')$ is a NO-instance. 

Now, to prove the reverse direction of the bidirectional statement, assume that $I_{j}(q',l_1', \ldots , l_t')$ is a NO-instance. Then, clearly $(U_{j-1}, q'-1, S_{1,j-1}, \ldots, S_{t,j-1}; l_1'-1, \ldots l_k'-1, l_{k+1}', \ldots, l_t')$ and $(U_{j-1}, q', S_{1,j-1}, \ldots, S_{t,j-1}; l_1', \ldots , l_t')$ are NO-instances, where 
$S_{i,j-1} = S_i \cap U_{j-1}$ for every $1\leq i\leq t$. Therefore, both $D_{j-1}(q'-1,l_1'-1, \ldots , l_k'-1, l_{k+1}', \ldots , l_t') = NULL$ and $D_{j-1}(q',l_1', \ldots , l_t') = \mbox{ NULL}$. Hence,  $D_j(q',l_1', \ldots , l_t') = \mbox{ NULL}$. This proves the theorem. 
\end{proof}

Below we state the main result of this section which follows directly from Theorems \ref{thm:TSS-MHS} and \ref{thm:multi_hitting_set} and from the fact that a polynomial-time $2$-factor approximation algorithm exists for computing a minimum vertex cover in a graph. 
\begin{theorem}
\label{thm:fpt_main_result}
Let $G$ be a graph on $n$ vertices with a threshold function $\tau:V(G)\rightarrow \mathbb{N}$ defined on its vertices. Let $t$ be the size of an optimal vertex cover in $G$. Then, the optimal target set for $G$ can be computed in time $2^{O(t \log t)}n^{O(1)}$.   
\end{theorem}


\section{FPT algorithm  parameterized by degree-1 modulator} \label{Sec:D1M}

In this section we prove the following theorem. 

\begin{theorem} \label{thm:TSS-D1Mmain}
\TSSDOM\ is solvable in time $2^{2^{O(t)}}n^{O(1)}$.  
\end{theorem}

As there is a simple branching algorithm of running time $O(3^t (n+m))$ to compute a degree-1 modulator, given a graph $G$ with a threshold function $\tau$, we can obtain an optimum target set in time $2^{2^{O(t)}}n^{O(1)}$, where $t$ is the size of a minimum degree-1 modulator of $G$.

As before we assume that 
the threshold of any vertex is at most its degree and this can be achieved using simple reduction rules. 
Also we assume that the minimum degree of a vertex in $G$ is at least $2$.
If there is a degree $1$ vertex $v$ in $G$ and its threshold is equal to $1$, then there is an optimum solution excluding $v$. Because if there is a solution $S$ containing $v$, then $(S\setminus \{v\})\cup N_G(v)$ is also a solution. 
Moreover, any solution for $G-v$ is also a solution for $G$.  
That is, we can do a simple reduction rule where we delete $v$ from the graph.  

Recall that for a degree-1 modulator $S$ in a graph $G$,
$G-S$ is a disjoint union of isolated vertices and isolated edges.

\begin{lemma}
\label{Lemma:D1diffusiontime}
Let $t$ be the size of an optimal degree-1 modulator in a graph $G$ with a threshold function $\tau:V(G) \rightarrow \mathbb{N}$. Then, the `diffusion process' starting from any non-empty seed set $S \subseteq V$ terminates in at most $3t$ rounds. 
\end{lemma}
\begin{proof}
Let $C$ be an optimal degree-1 modulator of $G$ of size $t$ and let $B=V(G)\setminus C$. Notice that  $G[B]$ is a collection of isolated vertices and isolated edges. Let $S_i$ be the set of uninfluenced nodes that were influenced in Round $i$ of the diffusion process. We have $S_0 = S$. 
Assume the diffusion process terminates in $\ell$ rounds. For each $0 \leq i \leq \ell$, observe that $S_i$ is a non-empty set. For $0 \leq i \leq \ell-2$, since $G[B]$ is a collection of isolated vertices and isolated edges, it is not possible to have both $S_i$ and $S_{i+2}$ to be subsets of $B$. Thus, in every three consecutive rounds, at least one uninfluenced vertex from $C$ will be influenced. This implies that $C\cap {\sf influence}_G(S)$ will be influenced in at most $3t-1$ steps. Therefore, the diffusion process will terminate in at most $3t$ steps. 
This proves the lemma. 
\end{proof}

%

Recall that by our assumption the minimum degree of a vertex in $G$ is at least $2$ and for any vertex 
$v\in V(G)$, $\tau(v)\leq d_G(v)$. 
We encode \TSSDOM\ as an {\sc Integer Programming (IP)} problem. In {\sc IP}, we are given $m$ linear constraints over $n$ variables and we want to check whether there is an integer assignment to the variables such that all the constraints are satisfied. More formally, the input consists of an $m\times n$ matrix ${\bf A}$ and a $m$-length column vector ${\bf b}$, and the objective is to test whether there exists an $n$-length vector ${\bf x}$ with all integer coordinates such that ${\bf Ax}\leq {\bf b}$.  

By the famous result of Lenstra\cite{Lenstra}, we know that {\sc ILP} parameterized by the number of variables is FPT. The current best known running time for solving an {\sc ILP} with $n$ variables and $m$ clauses is $2^{O(n \log n)} m^{O(1)}$~\cite{KannanIP}. 

\begin{proposition}[\cite{KannanIP}]\label{prop:Kannan}
{\sc IP} is solvable in time $2^{O(n\log n)} L\log L$, where $n$ is the number of variables and $L$ is the  input length. 
\end{proposition}

We reduce \TSSDOM\ to many instance of {\sc IP} such that the number of variables in each instances of {\sc IP} is bounded by a function of $t$ and \TSSDOM\ is a YES-instance if and only if at least one of the instances of {\sc IP} is a YES-instance.  Before giving a Turing reduction from  \TSSDOM\ to {\sc IP}, we prove some results which we use in the reduction. 

\begin{observation}
\label{obs:fullthr}
Let $G$ be a graph and $\tau:V(G) \rightarrow \mathbb{N}$ be a threshold function. Let $\{u,v\}$ be an edge in the graph $G$ such that $\tau(u)=d_G(u)$ and $\tau(v)=d_G(v)$. Then, for any target set, at least one vertex from $\{u,v\}$ belongs to $S$. 
\end{observation}
\begin{proof}
Consider the set $Q=V(G)\setminus \{u,v\}$. Since $\tau(u)=d_G(u)$ and $\tau(v)=d_G(v)$, both $u$ and $v$ will get influenced only after all of their neighbors get influenced. This implies that $Q$ is not a target set. 
Hence any subset of $Q$ is not a target set as well.  Thus, at least one vertex from $\{u,v\}$ should belongs to $S$ for any target set $S$.  
\end{proof}

\begin{lemma}
\label{lem:timestapforlarge}
There is an algorithm that given a graph $G$, a threshold function, $\tau:V(G) \rightarrow \mathbb{N}$, a degree-1 modulator $C$ and a time stamp $T\colon C \mapsto \mathbb{N}\cup \{0\}$ of an {\em unknown} target set $S$ (a hypothetical solution), runs in linear time, and outputs two functions $g \colon V(G)\setminus C\mapsto \mathbb{N}$ and $h \colon V(G)\setminus C\mapsto \mathbb{N}$ with the following properties. Let $u\in V(G)\setminus (C\cup S)$ and if $d_{G-C}(u)=1$, then $v$ is the only neighbor of $u$ in $G-C$. (Notice that $d_{G-C}(u)\leq 1$). 
\begin{itemize} 
\item[(ii)] If $d_{G-C}(u)=0$, then $g(u)= h(u)$ and 
$u$ is influenced on step $g(u)$ in the diffusion process starting from the target set $S$. 
\item[(ii)] If $d_{G-C}(u)=1$ and $v\in S$, $g(u)\leq h(u)$ and $u$ is influenced on step $g(u)$ in the diffusion process starting from the target set $S$. 
\item[(iii)] If $d_{G-C}(u)=1$ and $v\notin S$, then $g(u)\leq h(u)$ and $u$ is influenced on step $h(u)$ in the diffusion process starting from the target set $S$.
%
\end{itemize}  
\end{lemma}
\begin{proof}
Recall that minimum degree of a vertex in $G$ is at least $2$. 
Notice that we are given a time stamp $T\colon C \mapsto \mathbb{N}\cup \{0\}$. Let $B=V(G)\setminus C$. 
First we define the function $g$ as follows. 
Fix a vertex $u\in B$. Suppose $d_{G-C}(u)=0$. Then, $g(u)$ be the smallest positive integer $i$ such that $\vert \bigcup_{j<i} (T^{-1}(j)\cap N_G(u))\vert \geq \tau(u)$. Now suppose $d_{G-C}(u)=1$. Then, $g(u)$ be the smallest positive integer $i$ such that $\vert \bigcup_{j<i} (T^{-1}(j)\cap N_G(u))\vert \geq \tau(u)-1$.
It is easy to see that $g$ is well defined.

Next we define the function $h$. For a vertex $u\in B$, if $d_{G-C}(u)=0$, then $h(u)=g(u)$. 
Now for all vertices $u$ such that 
$d_{G-C}(u)=1$, we define $h$ iteratively such that the set $h^{-1}(i)$ is defined after defining $h^{-1}(j)$ for all $j<i$.  First we explain the vertices that will be mapped to $1$ by the function $h$. 
For each vertex $u\in B$ with $d_{G-C}(u)=1$, if $\vert T^{-1}(0)\cap N_G(u)\vert \geq \tau(u)$, then we define $h(u)=1$. 
Now let us consider an integer $i>1$ and assume that we have defined $h^{-1}(j)$ for all $j<i$.  Let $Y_{<i}=\bigcup_{j<i} (T^{-1}(j)\cup h^{-1}(j))$. Let $u$ be a vertex in $B\setminus Y_{<i}$ such that $d_{G-C}(u)=1$. 
We set $h(u)=i$ if $\vert Y_{<i} \cap N_G(u)\vert \geq \tau(u)$. If there is a vertex $u$ such that $h(u)$ is not defined so far, then we set $h(u)=\infty$. Clearly, $h$ is well defined. 
Since for any vertex $u\in B$, $d_{G-C}(u)\leq 1$, by the definition of $g$ and $h$, we have that $g(u)\leq h(u)$. Next, we prove that the properties $(i)$-$(iii)$ are satisfied.

\medskip
\noindent
{\bf Proof of property (i).} 
Here, we prove that the function $g$ and $h$ satisfies the property (i) mentioned in the lemma. 
Recall that $S$ is a target set and we want to prove that for all vertex $u\in B\setminus S$ with $d_{G-C}=0$, 
$g(u)=h(u)$ and  $u$ is influenced  on  step $g(u)$ in the diffusion process starting from $S$.  
Since $T$ is the time stamp of $S$ on $C$ and $N_G(u)\subseteq C$, $u$ is influenced on step $i$, where 
$i$ is the least integer such that $\vert \bigcup_{j<i} (T^{-1}(j)\cap N_G(u))\vert \geq \tau(u)$. By the definition of $g$ and $h$, we have that $g(u)=i$ and $h(u)=g(u)$.

\medskip
\noindent
{\bf Proof of property (ii).} 
Let $u\in B\setminus S$ with $d_{G-C}=1$ and $v$ be the only neighbor  of $u$ in $G-C$.  By our assumption, we have that $v\in S$. Thus, $u$ is influenced on step $i$, where 
$i$ is the least integer such that $\vert \bigcup_{j<i} (T^{-1}(j)\cap N_G(u))\vert \geq \tau(u)-1$. 
By the definition of $g$, we have that $g(u)=i$. 



\medskip
\noindent
{\bf Proof of property (iii).} 
We prove using induction on $i$ that any vertex $u\in B\setminus S$ such that $d_{G-C}(u)=1$ and 
$v\notin S$ (where $v$ is the only neighbor of $u$ in $G-C$), $u$ is influenced on step $i$ 
if and only if $h(u)=i$. The base case is when $i=1$. Let $u$ be a vertex in $B\setminus S$ such that 
$h(u)=1$, $d_{G-C}(u)=1$, and 
$v\notin S$ (where $v$ is the only neighbor of $u$ in $G-C$). 
Suppose, $h(u)=1$. By the definition of $h$, we have that  $\vert T^{-1}(0)\cap N_G(u)\vert \geq \tau(u)$. 
This implies that $u$ is influenced on step $1$.  
Now, for the other direction, suppose $u$ is influenced on step $1$. Thus, since $v\notin S$, we have that 
$\vert T^{-1}(0)\cap N_G(u)\vert \geq \tau(u)$. Therefore, by the definition of $h$, $h(u)=1$.

Now consider the induction step $i>1$. Let $u$ be a vertex in $B\setminus S$ such that 
$d_{G-C}(u)=1$, $v\notin S$ (where $v$ is the only neighbor of $u$ in $G-C$). 

$(\Rightarrow)$
Suppose 
$h(u)=i$. Let $j$ be the integer such that $v$ is influenced on step $j$ in the diffusion process starting from $S$. Consider the case when $j<h(u)$. Since $d_{G-C}(v)=1$, $v$ and $u$ do not belong to  $S$, by induction hypothesis we have that $h(v)=j$. Since $h(u)=i$, $i$ is the least integer such that $\vert Y_{<i} \cap N_G(u)\vert \geq \tau(u)$ where $\{v\}=(Y_{<i}\setminus C)\cap N_G(u)$. Since $T$ is a the time stamp on $C$ and $h(v)=j$, $u$ is influenced on step $i$. Now consider the case when $j\geq h(u)$. Then, by induction hypothesis, $h(v)\geq j$.  This implies that $\vert Y_{<i} \cap N_G(u)\vert=|\bigcup_{j<i} (T^{-1}(j)\cap N_G(u))|$, because $v\notin Y_{<i}$. Also, since $i$ is the least integer such that $\vert Y_{<i} \cap N_G(u)\vert \geq \tau(u)$, we have that $u$ is influenced on step $i$. 

$(\Leftarrow)$
Suppose $i$ be the integer such that $u$ is influenced on step $i$. Let $j=h(v)$. Consider the case when $j<h(u)$. By induction hypothesis $v$ is influenced on step $j$. Therefore, $i$ is the least integer such that 
$\vert Y_{<i} \cap N_G(u)\vert \geq \tau(u)$ where $\{v\}=(Y_{<i}\setminus C)\cap N_G(u)$. Hence, $h(u)=i$. Now consider the case when $j\geq h(u)$. By induction hypothesis $v$ is not influenced on or before step $h(u)-1$.  Therefore $i$ is the least integer such that $\vert \bigcup_{j<i} (T^{-1}(j)) \cap N_G(u)\vert \geq \tau(u)$. Since $v\notin Y_{<i}$,   $\vert Y_{<i} \cap N_G(u)\vert= \vert \bigcup_{j<i} (T^{-1}(j)) \cap N_G(u)\vert\geq \tau(u)$. Hence, $h(u)=i$. 
\qed
\end{proof}

\begin{theorem} \label{thm:TSS-D1M}
There is an algorithm that given an instance $(G, \tau, k, C)$ of  \TSSDOM, where $t=\vert C\vert$, runs in $2^{O(t \log t)}\cdot n^{O(1)}$ time, and outputs a collection of instances $\mathcal{I}$
 of {\sc IP} such that the following holds:
\begin{itemize}
\item (a) The number of instances is $2^{\mathcal{O}(t \log t)}$. 
\item (b) For every $I\in {\cal I}$, the number of variables in $I$ is upper bounded by $2^{O(t)}$. 
\item (c) $(G,\tau, k,C)$ is a YES-instance of \TSSDOM\ if and only if there exists $I\in {\cal I}$ such that 
$I$ is a YES-instance of {\sc IP}.  
\end{itemize}
\end{theorem}

\begin{proof}

Consider the given instance $(G, \tau, k, C)$ \TSSDOM, where $G$ is a graph on $n$ vertices and $C:=\{v_1, \ldots , v_t\}$ is a degree-$1$ modulator of $G$ of size $t$. Without loss of generality we assume that for each vertex $v\in V(G)$, $\tau(v)\leq d_G(v)$. Let $B=V(G) \setminus C$. 
We know that the maximum degree in $G[B]$ is at most $1$. That is $G[B]$ is a disjoint union of isolated vertices and edges. For each vertex $v_i$ in $C$, we guess a time stamp $T(v_i)$ in which $v_i$ will be influenced. From Lemma \ref{Lemma:D1diffusiontime}, we know that $T(v_i)\in \{0,1, \ldots , 3t\}$. 

There are $t$ vertices in $C$ and each one of them can be assigned any one of these $3t+1$ distinct values. So, there are $(3t+1)^t$ possible guesses for the time stamps of the vertices in $C$. Now, among all the $(3t+1)^t$ possibilities, let us consider the $j$-th one. That is, consider that we are given a function $T_j \colon C \mapsto \{0,1,\ldots,3t\}$ of guessed time stamps for vertices in $C$. Then, we apply Lemma~\ref{lem:timestapforlarge} on $G,C,\tau$ and $T_{j}$ and get two functions $g_{j},h_j\colon B \mapsto {\mathbb N}$.  
Now we will construct an instance $I_j$ of {\sc IP}, where the number of variables in $I_j$ is bounded by a function of $t$. Towards that we define an equivalence relation between the components of $G[B]$. Notice that the components of $G[B]$ are isolated vertices and isolated edges. For each edge $\{u,v\}$ in $G[B]$, we fix an arbitrary ordering $(u,v)$. Now, we define an equivalence relation $\sim$ on the components of $G[B]$ as follows. We say that for any two isolated vertices $x$ and $y$ in $G[B]$, $x\sim y$ if and only if $N_G(x)=N_G(y)$ and $\tau(x)=\tau(y)$. For any two components $e_1=(u_1,v_1)$ and $e_2=(u_2,v_2)$ in $G[B]$, $e_1\sim e_2$ if and only if $N_G(u_1)\cap C=N_G(u_2)\cap C$, $N_G(v_1)\cap C=N_G(v_2)\cap C$, $\tau(u_1)=\tau(u_2)$ and $\tau(v_1)=\tau(v_2)$. Notice that since $|C|=t$ and $C$ is a degree-1 modulator of $G$, for each $v\in B$, $d_{G}(v)\leq t+1$. This implies that the number of equivalence classes in $\sim$ is upper bounded by $2^{t}(t+1)+2^{2t}(t+1)^2\leq 2^{2t}(t+1)^3$. Let $\ell$ be the number of equivalence classes in $\sim$. We know that $\ell \leq 2^{2t}(t+1)^3$. 

Now, we are ready to construct the output {\sc IP} instance $I_j$ as follows. 
 We have a function $T_j \colon C \mapsto \{0,1,\ldots,3t\}$ 
that represents a time stamp on $C$ (by a hypothetical solution $S$),  and  functions $g_{j},h_j\colon B \mapsto {\mathbb N}$ satisfying the properties mentioned in Lemma~\ref{lem:timestapforlarge}. 


Next, we explain the set of variables in $I_j$. Let $F$ be an equivalence class in $\sim$. Notice that either all the elements in $F$ are isolated vertices or all the elements in $F$ are (ordered) isolated edges in $G-C$. 
If the elements in $F$ are isolated vertices, then we have one variable $x_F$ associated with it. 
We interpret the value assigned to $x_F$ by the {\sc IP} to be the number of vertices in $F$ included in the target set. 
If the elements in $F$ are isolated edges, then we have three variables $x_{F,1}$, $x_{F,2}$, and $x_{F,3}$ associated with $F$. The value $x_{F,1}$ represents the number of elements $(u,v)$ in $F$, where $u$ is included in the target set and $v$ is not included in the target set. Similarly, the value $x_{F,2}$ represents the number of elements $(u,v)$ in $F$, where $v$ is included in the target set and $u$ is not included in the target set. The value $x_{F,3}$ represents the number of elements $(u,v)$ in $F$, where both $u$ and $v$ are included in the target set.

Now we explain the constraints of $I_j$. 
Let ${\cal F}$ be the set of equivalence class in $\sim$.
Let ${\cal Q}$ be the equivalence classes in $\sim$ that contain isolated vertices and ${\cal P}={\cal F}\setminus {\cal Q}$.  
Notice that we are  looking for a target set of size at most $k$. 
Thus we have the following linear inequality. 
\begin{equation}
\sum_{F\in {\cal Q}}x_F + \sum_{P\in {\cal P}} (x_{P,1}+x_{P,2}+2x_{P,3}) \leq k-\vert T_j^{-1}(0)\vert \label{eq1}
\end{equation} 

Let $\{u,v\}$ be an edge in $G[B]$ such that $\tau(u)=d_G(u)$ and $\tau(v)=d_G(v)$. Without loss of generality assume that we have fixed the order $(u,v)$. Let $F$ be the equivalence class such that $(u,v)\in F$.  Note that 
for any $(u',v')\in F$, $\tau(u)=d_G(u)=\tau(u')=d_G(u')$ and $\tau(v)=d_G(v)=\tau(v')=d_G(v')$. 
Thus, by Observation~\ref{obs:fullthr}, any target set contains at least one vertex from $\{u,v\}$.
So, for any such equivalence class $F$, we have the following constraint. 

\begin{equation}
 (x_{F,1}+x_{F,2}+x_{F,3}) =\vert F\vert  \label{eq2}
\end{equation} 

From the proof of Lemma~\ref{lem:timestapforlarge}, we have that for an equivalence class $F\in {\cal Q}$ and $u,v\in F$, $g_j(u)=g_j(v)$ and $h_j(u)=h_j(v)$. Therefore, we slightly abuse the notation and use $g_j(F)$ and $h_j(F)$ to denote $g_j(u)$ and $h_j(v)$, respectively, where $u\in F$. 
Again from the proof of Lemma~\ref{lem:timestapforlarge}, 
we have that for an equivalence class $F'\in {\cal P}$ and $(u_1,v_1),(u_2,v_2)\in F'$, $g_j(u_1)=g_j(u_2)$, $h_j(u_1)=h_j(u_2)$, $g_j(v_1)=g_j(v_2)$, and $h_j(v_1)=h_j(v_2)$. 
Therefore, we slightly abuse the notation and use $g_j(F',1)$, $g_j(F',2)$, $h_j(F',1)$, and $h_j(F',2)$ to denote $g_j(u_1)$, $g_j(v_1)$, $h_j(u_1)$, $h_j(v_1)$, respectively, where $(u_1,v_1)\in F'$.

For each vertex $w\in C$ such that $T_j(w)\geq 1$, we construct a constraint. Towards that we need some notations. Let ${\cal Q}_w$ be the set of equivalence classes in ${\cal Q}$ such that 
for any vertex $v$ in an equivalence class in ${\cal Q}_w$, $\{v,w\}\in E(G)$.
Let ${\cal P}_{w,1}$ be the set of equivalence classes in ${\cal P}$ such that for any element $(u,v)$ in an equivalence class in ${\cal P}_{w,1}$, $\{u,w\}\in E(G)$, and $\{v,w\}\notin E(G)$.
Let ${\cal P}_{w,2}$ be the set of equivalence classes in ${\cal P}$ such that for any element $(u,v)$ in an equivalence class in ${\cal P}_{w,2}$, $\{u,w\}\notin E(G)$, and $\{v,w\}\in E(G)$.
Let ${\cal P}_{w,3}$ be the set of equivalence classes in ${\cal P}$ such that for any element $(u,v)$ in an equivalence class in ${\cal P}_{w,3}$, $\{u,w\},\{v,w\}\in E(G)$.
%
%
For a predicate $Z$, $[Z]$ returns $1$ if $Z$ is true and $0$ otherwise. 
The constraint we construct for $w$ is the following. 

\begin{eqnarray}
Y_w+Z_w+ \sum_{v\in N_G(w)\cap C} [T_j(v)<T_j(w)] \geq \tau(w) \label{eq3}
\end{eqnarray}
where $Y_w$ and $Z_w$ are defined as follows. 
\begin{eqnarray*}
Y_w&=& \sum_{F\in {\cal Q}_w}x_{F}+\sum_{F\in {\cal P}_{w,1}}(x_{F,1}+x_{F,3})+\sum_{F\in {\cal P}_{w,2}}(x_{F,2}+x_{F,3})+\sum_{F\in {\cal P}_{w,3}} (x_{F,1}+x_{F,2}+2x_{F,3})
\end{eqnarray*}
\begin{eqnarray*}
Z_w&=& \sum_{F\in {\cal Q}_w} [g_j(F)<T_j(w)](|F|-x_{F}) \\
&&+ \sum_{F\in {\cal P}_{w,1}} [g_j(F,1)<T_j(w)] x_{F,2}+ [h_j(F,1)<T_j(w)] (|F|-x_{F,1}-x_{F,2}-x_{F,3})\\
&&+ \sum_{F\in {\cal P}_{w,2}} [g_j(F,2)<T_j(w)] x_{F,1}+ [h_j(F,2)<T_j(w)] (|F|-x_{F,1}-x_{F,2}-x_{F,3})\\
&&+ \sum_{F\in {\cal P}_{w,3}} \Big([g_j(F,1)<T_j(w)] x_{F,2} +[g_j(F,2)<T_j(w)] x_{F,1}  \\
&&\qquad\qquad + ([h_j(F,1)<T_j(w)]+[h_j(F,2)<T_j(w)]) (|F|-x_{F,1}-x_{F,2}-x_{F,3})  \Big)\\
\end{eqnarray*}

%

%

\begin{equation}
\mbox{Also, for any equivalence class $F\in {\cal Q}$, we have }\quad 0\leq x_F \leq \vert F\vert  \label{eq4}
\end{equation} 

For any equivalence class $P\in {\cal P}$, we have 
\begin{eqnarray}
 (x_{P,1}+x_{P,2}+x_{P,3}) &\leq& \vert P\vert  \label{eq5}\\
 x_{P,1}, x_{P,2},x_{P,3} &\geq& 0 \label{eq6}
\end{eqnarray}

This completes the construction of the instance $I_j$. 
Recall that the number of equivalence classes is bounded by  $2^{2t}(t+1)^3$ and we have constructed at most $3$ variables per equivalence class. This implies that the number of variables in $I_j$ is at most 
$3\cdot 2^{2t}(t+1)^3$.  
Hence, condition (b) in the theorem is true. 
We have already mentioned that we constructed $(3t+1)^t$ instances. Thus condition (a) is satisfied.  The proof for condition (c) 
is moved to appendix. 
\end{proof}

Theorem~\ref{thm:TSS-D1M} and Proposition~\ref{prop:Kannan} imply Theorem~\ref{thm:TSS-D1Mmain}.

\section{Conclusion}
In this work, we improve the running time of \textsc{TSS-D$0$M} to 
$2^{{\cal O}(t\log t)}n^{{\cal O}(1)}$ by reducing the problem to \textsc{Multi\mbox{-}Hitting Set} and solving the latter. We believe that the \textsc{Multi\mbox{-}Hitting Set} problem could be of independent interest. In Section \ref{Sec:D1M}, we prove that \textsc{TSS-D$1$M} is FPT. An open question we propose here is the parameterized complexity of \textsc{TSS-D$2$M}.

%
%
%
%

\newpage

\subsection{Proof of Condition (c) in Theorem~\ref{thm:TSS-D1M}}

\begin{clm}[$\star$]
\label{clm:corrD1M}
Condition (c) is true. 
\end{clm}

\begin{proof}
Let $S$ be a solution to the instance $(G,\tau, k,C)$ of \TSSDOM. 
Then there exits $j\in [(3t+1)^t]$ such that $T_j \colon C \mapsto \{0,1,\ldots,3t\}$ is a time stamp on $C$ by the target set $S$. Recall that $g_j$ is the output of Lemma~\ref{lem:timestapforlarge} on the input $(G,\tau,C,T_j)$. Next we will show that $I_j$ is a YES-instance.  Towards that we define values for variables of $I_j$ and prove this is a solution for $I_j$.  

Let $F$ be an equivalence class in $\sim$. 
If the elements in $F$ are isolated vertices, then we set $x_F=\vert S\cap F \vert$. 
Suppose  the elements in $F$ are isolated edges.
Then, the values for $x_{F,1}$, $x_{F,2}$, and $x_{F,3}$ are as follows. The value for $x_{F,1}$ is the number of elements $(u,v)$ in $F$, where $u$ is included in $S$ and $v$ is not included in $S$. Similarly, the value for $x_{F,2}$ is the number of elements $(u,v)$ in $F$, where $v$ is included in $S$ and $u$ is not included in $S$. The value for $x_{F,3}$ is the number of elements $(u,v)$ in $F$, where both $u$ and $v$ are included in 
$S$. Since $\vert S\vert\leq k$, \eqref{eq1} is satisfied. By Observation~\ref{obs:fullthr}  and the fact that $S$ is a target set,  we have that \eqref{eq2} is satisfied. Let $F$ be an equivalence class in $\sim$ such that the elements in $F$ are singleton vertices. This implies that $N_G(x)=N_G(y)$ and $\tau(x)=\tau(y)$ for any $x,y\in F$. Thus, by the construction of $g_j$ in the proof of Lemma~\ref{lem:timestapforlarge}, we have that $g_j(x)=g_j(y)=g_j(F)$. Similarly, let $F'$ be an equivalence class in $\sim$ such that $(x_1,y_1),(x_2,y_2)$ be elements in $F'$. Then, again by the definition of equivalence classes in $\sim$ and the construction of $g_j$ and $h_j$, we have that $g_j(x_1)=g_j(x_2)=g_j(F',1)$, $g_j(y_1)=g_j(y_2)=g_j(F',2)$, $h_j(x_1)=h_j(x_2)=h_j(F',1)$, and $h_j(y_1)=h_j(y_2)=h_j(F',2)$ 

Now we prove that \eqref{eq3} is satisfied. 
Fix a vertex $w\in C$. Consider the L.H.S of \eqref{eq3}. Notice that $Y_w$ is equal to the number of neighbors of $w$ in $B\cap S$.  By Lemma~\ref{lem:timestapforlarge}, $Z_w$ is equal to the the number of neighbors of $w$ in $B\setminus S$ that are influenced strictly before the step $T_j(w)$. Thus, since $S$ is a target set and $w$ is influenced on step $T_j(w)$, \eqref{eq3} holds. 
%
%
%
%
It is easy to verify that \eqref{eq4}-\eqref{eq6} are satisfied. 
Thus, we conclude that $I_j$ is a YES-instance.

Now we prove the reverse direction of the proof. Suppose there is a YES-instance $I_j$. Then, we need to prove that $(G,\tau, k,C)$ is a YES-instance of \TSSDOM. 
As $I_j$ is a YES-instance, there are non-negative integer values for the variables $\{x_F ~\colon~ F\in {\cal Q} \} \cup \{x_{P,1},x_{P,2},x_{P,3}~\colon~P\in {\cal P}\}$ such that  \eqref{eq1}-\eqref{eq6} are satisfied. 
From this we construct a set $S$ of size at most $k$ and prove that $S$ is a target set. Towards the construction of $S$, initially set $S:=\emptyset$. Now consider an equivalence class $F\in {\cal Q}$. We arbitrarily choose $x_F$ vertices from the equivalence class $F$ and add to $S$. Now consider an equivalence class  $P\in {\cal P}$. We choose arbitrary pairwise disjoint subsets $P_1,P_2,P_3$ of $P$ such that $|P_1|=x_{P_1}$, $|P_2|=x_{P_2}$, and $|P_3|=x_{P_3}$. For each element $(u,v)\in P_1$, we add $u$ to $S$. For each element $(u,v)\in P_2$, we add $v$ to $S$. For each element $(u,v)\in P_3$, we add both $u$ and $v$ to $S$. Finally we add $T_j^{-1}(0)$ to $S$. This completes the construction of $S$.

Next, we prove that $S$ is a target set of size at most $k$. The construction of $S$ and \eqref{eq1} implies that $\vert S \vert \leq k$.  Because of \eqref{eq2} for any edge $\{u,v\}\in E(G[B])$ such that $\tau(u)=d_G(u)$
and $\tau(v)=d_G(v)$, at least one of $u$ or $v$ is in $S$. Because of this, from the proof of Lemma~\ref{lem:timestapforlarge}, for any $u\in B$, $h_j(u)\neq \infty$. 

Now we define a function  $c\colon B\setminus S \mapsto {\mathbb N}$ as follows. For a vertex $u\in B\setminus S$, if $d_{G-C}(u)=0$, then we define $c(u)=g_j(u)$. 
Consider a vertex $u\in B\setminus S$ such that $d_{G-C}(u)=1$. Let $v$ be the only neighbor of $u$ in $G-C$.  Then, if $v\in S$, then  we define $c(u)=g_j(u)$ and otherwise $c(u)=h_j(u)$. Since, for any $u\in B$, $h_j(u)\neq \infty$, we have that $c(u)\neq \infty$.

Next using induction on $i$, we prove that the vertices in the set $T^{-1}_j(i)\cup c^{-1}(i)$ are influenced on step $i$ in the diffusion process starting from $S$.  The base case is when $i=1$. Let $w$ be a vertex in $T_j^{-1}(1)$. Consider \eqref{eq3}. Note that $Z_w=0$. Thus, we have that $Y_w+
\sum_{v\in N_G(w)\cap C} [T_j(v)<1] \geq \tau(w)$. By the construction of $S$, there are $Y_w$ vertices in $S\cap B$ that are neighbors  of $w$. Therefore, $w$ is influenced on step $1$. Let $w'\in c^{-1}(1)$. Since 
$T_j^{-1}(0)$ are influenced on step $0$ in the diffusion process starting at $S$, by Lemma~\ref{lem:timestapforlarge}, $w'$ is influenced on step $1$ (because in a diffusion process, the vertices influenced on step $1$ depends only on the seed set). 



Now consider the induction step $i>1$. Let $S_{<i}$ be the set of vertices influenced before step $i$. 
The construction of $S$, \eqref{eq3} and the induction hypothesis imply that for any $w\in T_j^{-1}(i)$, $\vert S_{<i}\cap N_G(w)\vert \geq \tau(w)$. Now, consider a vertex $w'\in c^{-1}(i)$.  By the induction hypothesis 
for any vertex $x\in \bigcup_{r<i} T^{-1}_j(r)$, $x$ is influenced on step $T_j(x)$. Thus, by Lemma~\ref{lem:timestapforlarge} and the definition of the function $c$, $w'$ is influenced on step $i$ (because in a diffusion process, the vertices influenced on step $i$ depends only on the vertices influenced before step $i$). 

Since $c(u)\neq \infty$ for any $u\in B\setminus C$, we have that all the vertices are influenced in the diffusion process starting from $S$. This completes the proof of the claim.
%
%
%
\end{proof}

\end{document}